\documentclass[12pt]{amsart}

\usepackage{amsfonts,amssymb,mathrsfs}
  \usepackage{amsthm}

\newtheorem{theorem}{Theorem}[section]
\newtheorem{proposition}[theorem]{Proposition}
\newtheorem{ex}[theorem]{Example}

\newtheorem{lemma}[theorem]{Lemma}
\newtheorem{definition}[theorem]{Definition}
\newtheorem{cor}[theorem]{Corollary}
\newtheorem{conj}[theorem]{Conjecture}

\newcommand\ep{{\varepsilon}}

\newcommand\de{\delta}
\newcommand\zu{[0,1]}
\newcommand\zud{[0,1]^d}
\newcommand\zd{[0,d]}

\newcommand \bfk{{\bf k}}
\newcommand \bfe{{\bf e}}

\newcommand{\R}{\mathbb{R}}
\newcommand{\N}{\mathbb{N}}
\newcommand{\Z}{\mathbb{Z}}

\newcommand{\calm}{\mathcal{M}}

\newcommand{\calr}{\mathcal{R}}
\newcommand{\calc}{\mathcal{C}}

\newcommand{\ho}{H\"older }

\newcommand{\mk}{\medskip}
\newcommand{\sk}{\smallskip}

\newcommand{\DDD}{\Delta}

\newcommand{\caa}{{\mathcal A}}

\newcommand{\cac}{{\mathcal C}}

\newcommand{\cag}{{\mathcal G}}

\newcommand{\ddd}{\delta}

\newcommand{\defeq}{{\buildrel {\rm def}\over =}}

\newcommand{\eee}{\epsilon}

\newcommand{\intt}{\mathrm{int}}

\newcommand{\mmm}{\mu}

\newcommand{\oB}{\overline{B}}

\newcommand{\oo}{\infty}

\newcommand{\rrr}{\rho}
\newcommand{\rrrr}{\varrho}
\newcommand{\sm}{\setminus}
\newcommand{\sse}{\subset}

\newcommand{\ttt}{\theta}
\newcommand{\tttt}{\tau}

\begin{document}

\title[Typical Borel measures on $\zu^d$ satisfy a multifractal formalism] {Typical Borel measures on $\zu^d$ satisfy a multifractal formalism}

\author{Zolt\'an Buczolich}

\address{Zolt\'an Buczolich - Department of Analysis, E\"otv\"os Lor\'and
University, P\'azm\'any P\'eter S\'et\'any 1/c, 1117 Budapest, Hungary }

\author{St\'ephane Seuret }

\address{St\'ephane Seuret  - LAMA, CNRS UMR 8050, Universit\'e Paris-Est Cr\'eteil Val-de-Marne,  
  61 avenue du G\'en\'eral de Gaulle, 
94 010 CR\'ETEIL Cedex, France}



\begin{abstract}
In this article, we prove that in the Baire category sense, measures supported by    the unit cube of $  \R^d$ typically satisfy a multifractal formalism. To achieve this, we compute explicitly  the multifractal spectrum of such typical measures $\mu$. This spectrum appears to be linear  with slope $1$,   starting from 0 at exponent 0, ending at dimension $d$ at exponent $d$, and it indeed coincides with the Legendre transform of the $L^q$-spectrum associated with typical measures $\mu$. 
\end{abstract}

\keywords{Borel measures, Hausdorff dimension, Multifractal analysis, Baire categories.}

\maketitle 
\section{Introduction}

 Let $\calm(\zud) $ be the set of probability measures on $\zu^d$
endowed with the weak topology (which makes  $\calm(\zud) $   a compact separable  space). 
 Recall that 
the local regularity of a positive measure $\mu \in \calm(\zud) $  at a given $x_0\in\zu$ is quantified by a {\em local dimension} (or a {\em local \ho exponent}) $h_\mu(x_0)$, defined as
\begin{equation}
\label{defexpmu}
h_\mu(x_0)=\liminf_{r\to 0^+}  \frac{\log \mu (B(x_0,r) )} { \log r},
\end{equation}
where $B(x_0,r)$ denotes the ball with center $x_0$ and radius $r$. 
In geometric measure theory $h_{\mmm}(x_{0})$ is called
the lower local dimension of $\mmm$ at $x_{0}$ and is denoted by
$\underline{\dim}_{\mathrm{loc}}\mmm(x_{0})$.
Then the singularity spectrum of $\mu$ is the map 
$$d_\mu: h\geq 0 \mapsto \dim_{\mathcal{H}}\,  E_\mu(h),$$
$\mbox{where }$
\begin{equation}
\label{defemuh}
  E_\mu(h):=\{x\in\zud: h_\mu(x)=h\}.
  \end{equation}

This spectrum describes the distribution of the singularities of the measure $\mu$, and thus contains crucial information on the geometrical properties of $\mu$. Most often, two forms of spectra are obtained for measures: either a spectrum with the classical concave shape (obtained as Legendre transform of some concave $L^q$-scaling function, for instance in the case of self-similar measures, Mandelbrot cascades and their extensions, see  \cite{Ba1,BM2,BRMICHPEY,OLSEN1,FALC,PESIN,RAND} for historical references, among many references), or a linear increasing spectrum (as in \cite{AB,SEUILLAGE1,JAFFARD1}). 

These two distinct shapes arise in different contexts: On one hand, linear spectra are usually found for measures and functions which are infinite sums of mutually independent contributions, i.e. which are obtained from an {\em additive} procedure.    L\'evy subordinators, which are integrals of  infinite sum of randomly distributed Dirac masses,  and random wavelet series, where the wavelet coefficients are i.i.d. random variables, illustrate this fact.  For such stochastic processes, the greatest \ho exponent coincides with the almost sure exponent, meaning that at Lebesgue almost every point, the sample path of the process enjoys the highest possible local regularity. On the other hand, concave spectra are generally  obtained for measures or functions built using a {\em multiplicative} or {\em hierarchical} scheme. As said above,   Mandelbrot cascades are the archetypes of measures with a multiplicative structure and exhibit in full generality a concave spectrum.  In such constructions, the strong local correlations make it possible the presence of points around which the local exponent is greater than the almost sure exponent. This constitutes a striking difference with additive processes, for which these more regular points do not exist. 

Subsequently, the shape of the spectrum may reflect the structure of the object under consideration, and may reveal some properties of the physics underlying the signal, if any. Hence, it is very natural to investigate the structure of typical measures. Actually we will prove that typical measures tend to exhibit an {\em additive} structure, and the proof we develop will exploit this property.

\mk

Before stating our result, we recall the notion of $L^q$-spectrum for a probability measure $\mu \in \calm(\zud)$.
  If $j$ is an integer greater than $1$,  then we set
  \begin{eqnarray}
  \label{defzj}
\Z_j  =    \{0,1,\cdots, 2^j-1\} ^d .
 \end{eqnarray}
Then, 
 let $\mathcal{G}_j$ be the partition of $[0,1)^d$ into dyadic  boxes: $\mathcal{G}_j$ is the set of all cubes
 \begin{eqnarray*}
   I_{j,\bfk}& \defeq &  \prod_{i=1}^d  \ [k_i2^{-j},  (k_i +1)2^{-j}),
 \end{eqnarray*}
 where $\bfk:= (k_1,k_2,\cdots,k_d)\in \Z_j$.
 
 The $L^q$-spectrum of a measure $ \mu \in \calm(\zud)$ is the mapping defined for any   $q\in\mathbb{R}$ by
\begin{equation}
\label{deftau}
\tau_ \mu (q)=\liminf_{j\to\infty}-\frac{1}{j}\log_2 s_j(q)\ \ \mbox{ where } \ \ s_j(q)= \! \! \! \sum_{Q\in\mathcal{G}_j,\   \mu (Q)\neq 0}  \mu  (Q)^q .
\end{equation}


It is classical  \cite{BRMICHPEY,OLSEN1} that the Legendre transform of $\tau_\mu$ serves as upper bound for the multifractal spectrum $d_\mu$: For every $h\geq 0$, 
\begin{equation}
\label{maj1}
d_\mu(h) \leq  (\tau_\mu)^*(h) := \inf _{q\in \R} (qh -\tau_\mu(q)).
\end{equation}

A lot of work has been achieved  to prove that for specific measures (like self-similar measures, ...,  see all the references above)  the upper bound in \eqref{maj1} turns out to be  an equality. When   \eqref{maj1} is an equality at exponent $h\geq 0$,  the measure is said to {\em satisfy the multifractal formalism} at   $h$. The validity of the multifractal formalism for given measures is a very important issue in   Mathematics and in  Physics, since  when it is known to be satisfied,  it makes it possible to estimate the singularity spectrum of real data through the estimation of the $L^q$-spectrum.  Moreover, it  gives important information   on the  geometrical properties (from the viewpoint of geometric measure theory)  of the measure $\mu$ under consideration. 

\mk

These considerations led us also to find out whether  the validity of the multifractal formalism is  {\em typical} (or {\em generic}). 
Recall that a  property is said to be typical in a complete metric space $E$, when it holds on 
a residual set, i.e. a set  with a complement of first Baire category. A set is of first Baire category if it is the 
union of countably many nowhere dense sets. Most often, including in this paper,  one can verify that the
residual set is dense $G_{\de}$, that is, a countable  intersection of dense open sets in $E$. 

A first result in this direction  was found by
Buczolich and Nagy, who proved in \cite{BucNag} that typical continuous probability measures on $\zu$ have   a linear increasing spectrum with slope 1, and satisfy the formalism. Then Olsen   studied the typical $L^q$-spectra of measures on general compact sets \cite{OLSEN3,OLSEN4}, but did not compute the multifractal spectrum of typical measures.

\mk

In this paper,   we are interested in  the form of the multifractal spectrum of typical Borel measures in $\calm(\zud) $, and we investigate whether the multifractal formalism is typically satisfied for such measures.

\begin{theorem}
\label{th1}
There is a dense $G_\delta$ set $\mathcal{R}$ included in  $\calm(\zud)$ such that for every measure $\mu \in \calr$, we have
\begin{equation}\label{*th11*}
\forall \ h\in \zd, \ \ \ \ d_\mu(h)=h,
\end{equation}
and $E_\mu(h)=\emptyset$ if $h>d$.

In particular, for every $q\in [0,1]$, $\tau_\mu(q)=d(q-1)$, and $\mu$ satisfies the multifractal formalism at every $h\in [0,d]$, i.e. $d_\mu(h) = \tau_\mu^*(h)$.
 \end{theorem}
 
 We note that there is a slight difference in notation in \cite{BucNag}
since in \eqref{deftau} there is a negative sign in the definition
of $\tttt_{\mmm}(q)$.
 Since we compute the multifractal spectrum of typical measures $\mu$, using \eqref{maj1}, we recover part of the result of Olsen \cite{OLSEN4}, i.e.  the value of $\tau_\mu(q)$ of $q\in \zu$, when the support of the measure is $\zu^d$. 
 
 \mk
 
 We conjecture that similar properties hold on all compact sets of $\R^d$.
 
\mk\noindent
\begin{conj}
  For any compact set $K \subset \R^d$, there exists a constant $0 \leq D \leq d$ such that typical measures  $\mu$ (in the Baire sense) in $\calm (K)$ satisfy: for every $h\in [0,D]$, $d_\mu(h)=h$, and if $h>D$, $E_\mu(h)=\emptyset$.
\end{conj}
 
 \mk
 
 Whether $D$ should be  the Hausdorff dimension of $K$ or the lower box dimension of $K$ (or another dimension)  is not obvious for  us at this point.
 
 \mk

In the rest of this work,   pure atomic   measures of the form ($\de_x$ stands for the Dirac measure at $x\in \zu^d$)
\begin{equation}
\label{defnu0}
\nu = \sum_{n\geq 0} \ r_{n}  \cdot \de_{x_n},
\end{equation}
will play a major role. For instance, the separability of $\calm(\zud) $ follows from the fact that measures $\nu$ of the form \eqref{defnu0}, where $(r_{n})_{n\geq 0}$ are positive rational numbers such that  $\sum_{n\geq 0} \ r_{n} =1$,  and where $(x_n)_{n\geq 0}$ are rational points of the cube $\calm(\zud)$, form a  countable dense set in $\calm(\zud)$ for the weak topology. 
Atomic measures $\nu$   have been studied  by many authors \cite{AB,BS1,BS2,SEUILLAGE1,SEUILLAGE2,JAFFARD1,OLSEN1}. In particular, it is shown in \cite{SEUILLAGE1,SEUILLAGE2}  that such measures always exhibit specific multifractal properties.

\mk

 The paper is organized as follows. Section \ref{sec_prel} contains the  precise definitions and some known  results  on dimensions and  multifractal spectra for Borel measures, as well as some recalls on the properties of $\calm(\zud)$.  We also prove the seond part of Theorem \ref{th1}, i.e.  for generic measures, $\tau_\mu(q)=d(q-1)$ for every $q\in \zu$.
 
 In Section \ref{sec_cons0}, we build a  dense $G_\de$ set $\calr$ of measures in  $\calm(\zud)$ such that for every $\mu\in \calr$, {\em for every $x\in \zud$}, $h_\mu(x)\leq d$ and {\em for Lebesgue-almost every $x\in \zud$}, $h_\mu(x)=d$.
 
In Section \ref{sec_cons2} we prove  that  for every $\mu\in\calr$, for every $h\in [0,d)$, $d_\mu(h)=h$.
This implies Theorem \ref{th1}. 
\section{Preliminary results}
\label{sec_prel}

In $\R^{d}$ we will use the metric coming from the supremum norm,
that is, for $x,y\in\R^{d}$, $\rrr(x,y)=\max\{ |x_{i}-y_{i}|:i=1,...,d \}.$

 The open ball centered at $x$ and of radius $r$ is denoted by $B(x,r)$.
The closure of the set $A\subset \R^d$ is denoted by $\overline{A}$, 
moreover
$  |A|$ and  $\mathcal{L}_d(A)$
denote its diameter and $d$-dimensional Lebesgue measure, respectively.

\subsection {Dimensions of sets and measures}

We refer the reader to \cite{FALC} for the standard definition of Hausdorff measures $\mathcal{H}^s(E)$ and Hausdorff dimensions $\dim_{\mathcal{H}}(E) $ of a set $E$.

For a Borel measure $\mu\in \calm(\zud)$, one   defines  the dimension of $\mu$ as
\begin{equation}\label{*dimmu}
 \dim_{\mathcal{H}}( \mu ):=\sup\{ s: h_\mu(x)\geq s \mbox{ for $\mu$-a.e. } x\}. 
\end{equation}
 By Proposition 10.2 of \cite{FALC}
 $$\dim_{\mathcal{H}}(\mu) = \inf\{\dim_{\mathcal{H}} (E): E\subset \zu \mbox{ Borel and } \mu(E)>0\}.$$
 The following property will be particularly relevant: 
 \begin{eqnarray}
 \label{property}
 \begin{array}{c}
 \mbox{if $\dim_{\mathcal{H}}(\mu)\geq h$, then for every Borel set $E\subset \zu$}\\
 \mbox{of dimension strictly less than $h$,  $\mu(E)=0$.}
 \end{array}
 \end{eqnarray}

%


We recall standard results on multifractal spectra of Borel probability measures. 

\begin{proposition}
\label{prop0}
Let $\mu\in\calm(\zud)$ and 
\begin{equation}
\label{embed}
\widetilde {E_\mu}(h) = \{  x\in\zud: h_\mu(x)\leq h\}\supset E_{\mmm}(h).
\end{equation}

For every $h \geq 0$, $d_{\mmm}(h)=\dim_{\mathcal{H}} E_{\mmm}(h)\leq
\dim_{\mathcal{H}} \widetilde {E_\mu}(h) \leq \min(h,d)$.

\end{proposition}

This follows for instance from  Proposition 2.3 of \cite{[FaT]}, where it is shown that  for 
$$\widetilde {E_\mu}(h) = \{  x\in\zud: h_\mu(x)=
\underline{\dim}_{{\mathrm {loc}}}\mmm(x)
\leq h\}$$ we have $\dim_{\mathcal{H}} \widetilde {E_\mu}\leq h$. The rest follows from the embedding \eqref{embed}.

\mk

From this we deduce in Theorem \ref{th1} that for typical measures, $\tttt_{\mmm}(q)=d(q-1)$
for all $q\in [0,1]$. We prove it quickly for completeness.

\begin{cor}
Assume that  \eqref{*th11*} holds true for a  probability measure $\mu$ on $\zud$. Then $\tttt_{\mmm}(q)=d(q-1)$
for all $q\in [0,1]$. 
\end{cor}
\begin{proof}
Recall that $\tau_\mu$ and $s_j(q)$ were defined in \eqref{deftau}. Since $\cag_{j}$ has $2^{dj}$ many cubes in $\zud$ by using H\"older's inequality for $0<q<1$
$$s_{j}(q)\leq (\sum_{Q\in \cag_{j}} \mu(Q)^{q/q})^{q} (\sum_{Q\in\cag_{j}}1^{1/(1-q)})^{1-q}=1\cdot (2^{jd})^{1-q}.$$
This implies $\tttt_{\mmm}(q)\geq d(q-1)$.  One could also  notice that $\tau_\mu(0)=-d$, $\tau_\mu(1)=0$ and $\tau_\mu$ is a concave map on the interior of its support and hence $\tttt_{\mmm}(q)=d(q-1)$
for all $q\in [0,1]$.

Assume now  that   \eqref{*th11*} holds true for  $\mu$.  
Proceeding towards a contradiction suppose that
there exists $q'\in (0,1)$ such that $\tttt_{\mmm}(q')> d(q'-1)$.
By concavity of $\tttt_{\mmm}(q)$ there exists  $d'<d$
such that $\tttt_{\mmm}(q)>d'(q-1)$ for all $q\in (q',1)$.
Hence for $d'<h<d$ by \eqref{maj1}  and \eqref{*th11*} we have
$$h=d_{\mmm}(h)\leq \inf_{q\in\R}(qh-\tttt_{\mmm}(q))\leq \inf_{q\in (q',1)}(qh-d'(q-1))=
\inf_{q\in (q',1)} (q(h-d')+d')<h,$$
a contradiction.
This concludes the proof.
\end{proof}

\subsection {Separability of $\calm(\zud)$}

Let us denote by  Lip$(\zu^d)$ the set of Lipschitz functions on $\zu^d$ with Lipschitz constant 
not exceeding
$1$. 
Recall that the weak topology on $\calm(\zud)$ is induced by the following metric: if $\mu$ and $\nu$ belong to $\calm(\zud)$, we set
\begin{equation}
\label{defmetric}
\rrrr(\mu,\nu)= \sup \left\{ \left|\int f\, d\mu -\int f\, d \nu  \right|: f\in \mbox{Lip}(\zu^d) \right\}.
\end{equation}

As is mentioned in the introduction,   $\calm(\zud)$ is  a separable set. For our purpose, we specify a countable dense family of atomic measures. Indeed, the set of finite atomic measures of the form 
\begin{equation}
\label{defnu1}
  \sum_{{\bf k}\in \Z_j } \ r_{j,{\bf k}}  \cdot \de_{{\bf k} 2^{-j}},
\end{equation}
where $j\in \N^*=\N\setminus\{ 0 \} $, $\Z_j$ was defined  by \eqref{defzj}  and  $(r_{j,{\bf k}})_{{\bf k}\in  \Z_j}$ are  
(strictly) positive  rational numbers such that  
$$ \sum_{{\bf k}\in \Z_j } \ r_{j,{\bf k}} =1,$$
 forms a dense set in $\calm(\zud)$ for the weak topology.

\section{The construction of $\calr$, our dense $G_\de$ set in $\calm(\zud)$}
\label{sec_cons0}
  
We build a dense $G_\de$ set $\calr$ in $\calm(\zud)$.
In this section we show 
that for every $\mu\in\calr$, for every $x\in \zud$, $h_\mu(x)\leq d$, and for Lebesgue-almost every $x\in\zud$, $h_\mu(x)=d$.

\mk

Let us  enumerate the measures of the form \eqref{defnu1} as a sequence $\{\nu_1$, $\nu_2$, ..., $\nu_n$,..., $\}$.

  Let $n\geq 1$, and consider $\nu_n$. We are going to construct another measure $\mu_n$, close to $\nu_n$ in the weak topology, such that $\mu_n$ has a very typical behavior at a certain scale.

  Let us write the measure $\nu_n $ as  
\begin{equation}
\label{defnun}
\nu_n =  \sum_{{\bf k}\in \Z_ {j_n}  } \ r_{j_n,{\bf k}}  \cdot \de_{{\bf k} 2^{-j_n}},
\end{equation}
  where $j_n$ is the  integer such that $r_{j_n,{\bf k}} >0$ for all $\bfk \in  \Z_ {j_n}$ ($j_n$ is necessarily unique since all  the Dirac masses in  measures in \eqref{defnu1} have   a strictly positive weight).

 Set $\bfe=(1/2,...,1/2)\in \R^{d}$.

For every integer $j\geq 1$, let us introduce the measure $\pi_j$ defined as
\begin{equation}
\label{defpij}
\pi_j = \sum _{\bfk \in  \Z_ {j} } \ 2^{-dj} \, \de_{(\bfk+\bfe) 2^{-j}}.
\end{equation}
This measure $\pi_j$ consists of Dirac masses located at 
centers of 
the dyadic cubes of $\zud$ of generation $j$, and gives the same weight
 to each Dirac mass. 
For every integer $n\geq 1$, let 
\begin{equation}
\label{defkn}
J_n =    2n  ( j_n)^2   ,
\end{equation}
  so that $J_n/n \geq 2j_n$.  Finally, for every $n\geq 1$, we define
\begin{equation}\label{*mnddef*}
\mu_n = 2^{-J_n/n} \cdot  \pi_{J_n} + (1-2^{-J_n/n})  \cdot \nu_n.
\end{equation}
  Obviously, for every $\bfk\in  \Z_ {J_n} $, we have
  \begin{equation}
\label{inegcrucial}
\mu_n(I_{J_n,\bfk}) \geq 2^{-J_n/n} \cdot  \pi_{J_n} (I_{J_n,\bfk})
 \geq 2^{-J_n/n}2^{-d     J_n}  =  |I_{J_n,\bfk}|^{d+1/n}
\end{equation}
where the last equality holds since we use the supremum metric.

  \begin{lemma}
  For every $n\geq 1$, $\rrrr(\mu_n,\nu_n) \leq  2\cdot 2^{-J_n/n }$.
  \end{lemma}
  \begin{proof}
  Recall Definition \eqref{defmetric} of the metric $\rrrr$. Let $f\in$  Lip$(\zu^d)$.  We have 
  \begin{eqnarray*}
  \left|\int f \ d \nu_n- \int f \ d \mu_n \right|  & =  & 2^{-J_n/n} \left|\int f \ d \nu_n -\int f \ d \pi_{J_n} \right| \\
   & = & 2^{-J_n/n}  \rrrr( \nu_n , \pi_{J_n}) \leq 2\cdot 2^{-J_n/n }.  \end{eqnarray*}
       \end{proof}

The density of the sequence $(\nu_n)_{n\geq 1}$ implies the density of $(\mu_n)_{n\geq 1}$, since the distance $\rrrr(\mu_n,\nu_n)$ converges to zero as $n$ tends to infinity.

   \begin{definition}
  We introduce  for every $N\geq 1$ the  sets in $\calm(\zud)$
  \begin{equation}
  \label{defR}
\Omega_N ^{\rrrr}= \bigcup_{n\geq N} \ \   {B}(\mu_n,   2^{-(d+4) (J_n)^{2}})
  \ \mbox{ and } \ 
  \mathcal{R} = \bigcap_{N\geq 1} \ \ \Omega^\rrrr_N,
  \end{equation}
  where the open balls are defined using  the metric $\rrrr$ defined by \eqref{defmetric}.

  Each  set $\Omega^\rrrr_N$ is obviously a dense open set in $\calm(\zud)$, hence $\mathcal{R}$ is a dense $G_\de$ set in $\calm(\zud)$.
  \end{definition}

  %


%


\section{Upper bound for the local \ho exponents of typical measures}
\label{sec_cons2}

We first prove that all exponents of typical measures  $\mu$ are less than $d$, and then, in the last section,  we compute the whole spectrum of $\mu$.

\begin{proposition}
\label{propmajorationd}
For every $\mu \in \calr$, for every $x\in \zud$, $h_\mu(x)\leq d$.
\end{proposition}
\begin{proof}
Let  $\mu \in \calr$. There is a sequence of positive integers $( {N_p})_{p\geq 1}$ growing to infinity such that for every $p\geq 1$, $\rrrr(\mu,\mu_{N_p} ) \leq 2^{-(d+4)(J_{N_p})^{2} }$.

Suppose that $\bfk\in  \Z_ {J_{N_p}}$.
We introduce an auxiliary function $f_{{p},\bfk}$ defined as follows: First we set $f_{{p},\bfk}((\bfk+\bfe)2^{-J_{N_{p}}})=2^{-J_{N_{p}}-1}$
and $f_{{p},\bfk}(x)=0$ for $x\not\in I_{J_{N_p},\bfk}.$
Then  we use an extension of $f_{{p},\bfk}$
onto $I_{J_{N_p},\bfk}$
such that
$f_{{p},\bfk}\in \mbox{Lip}(\zud)$ and $0\leq f_{{p},\bfk}\leq 2^{-J_{N_{p}}-1}.$ 

First observe that 
 \begin{eqnarray}\int f_{{p},\bfk} \, d \mu & \leq  & 2^{-J_{N_{p}}-1}\mmm(I_{J_{N_p},\bfk}).
 \label{eq00}
   \end{eqnarray}

Moreover, we have 
  \begin{eqnarray}
\label{eq01}
 \int f_{{p},\bfk} \, d \mu_{N_p} & \geq  & 2^{-J_{N_{p}}/N_{p}}\int 
f_{{p},\bfk}   \, d\pi_{J_{N_{p}}}\\  \nonumber
 & \geq & 2^{-J_{N_{p}}/N_{p}} 2^{-dJ_{N_{p}}}\int 
f_{{p},\bfk} \, d \ddd_{(\bfk+\bfe)2^{-N_{p}} } \\
& \geq &  \nonumber
2^{-J_{N_{p}}/N_{p}} 2^{-dJ_{N_{p}}} 2^{-J_{N_{p}}-1}\\
& = &  
|I_{J_{N_p},\bfk}|^{d+1/N_{p}}\cdot 2^{-J_{N_{p}}-1}=
2^{-J_{N_{p}}(d+1+1/N_{p})-1}.  \nonumber
  \end{eqnarray}
We also have
  \begin{eqnarray}
  \label{eq02} \rrrr(\mu,\mu_{N_p}) \leq 2^{-(J_{N_p})^{2}(d+4) } < 2^{-J_{N_{p}}(d+1+1/N_{p})-2}.
    \end{eqnarray}
Combining \eqref{eq00}, \eqref{eq01}  and \eqref{eq02},  we deduce that  
  \begin{eqnarray*} 
  2^{-J_{N_{p}}-1}\mmm(I_{J_{N_p},\bfk})  & >   & 2^{-J_{N_{p}}(d+1+1/N_{p})-1}-  2^{-J_{N_{p}}(d+1+1/N_{p})-2}\\
   &  > & 2^{-J_{N_{p}}(d+1+1/N_{p})-2} , 
    \end{eqnarray*}
 which leads to  
$$\mmm(I_{J_{N_p},\bfk})> 2^{-J_{N_{p}}(d+1/N_{p})-1}=
|I_{J_{N_p},\bfk}|^{d+1/N_{p}+1/J_{N_{p}}}>
|I_{J_{N_p},\bfk}|^{d+2/N_{p}}.$$

For any integer $j\geq 1$, let us denote by $I_j(x)$ the unique dyadic cube of generation $j$ that contains $x$. Recalling the definition of the \ho exponent \eqref{defexpmu} of $\mu$ at  any $x\in \zu^d$, we obviously obtain 
\begin{eqnarray*} h_\mu(x) & = & \liminf_{r\to 0^+}  \frac{\log \mu (B(x,r) )} { \log r}  \leq \liminf_{ p \to +\infty}  \frac{\log \mu (I_{J_{N_p}}(x) )} { \log |I_{J_{N_p}}(x)|}\\
& \leq  &  \liminf_{ p \to +\infty} d+2/N_p =d.\vspace{-2mm}
\end{eqnarray*}
\vspace{-3mm}\end{proof}

\section{The multifractal spectrum of typical measures of $\calr$}
\label{sec_cons3}

Let $\mu\in \mathcal{R}$, where $\mathcal{R}$ was defined by \eqref{defR}.  
 
Hence,  there is a sequence of integers $(N_p)_{p\geq 1}$ such that  $\mu \in \Omega_{N_p}^{\rrrr}$ for every $p\geq 1$. Equivalently, for every $p\geq 1$,  $\rrrr(\mu, \mu_{N_p}) \leq   2^{-(J_{N_p}
)^2(d+4)}$, where $\mu_{N_p}$ is given by \eqref{*mnddef*}.

\mk

We are going to prove that such a measure $\mu$ has necessarily a multifractal spectrum equal to $d_\mu(h)=h$, for every $h\in [0,d)$. Recall that we already have the upper bound $d_\mu(h) \leq h$, hence it suffices to bound from below the Hausdorff dimension of each set $\{x\in \zud:h_\mu(x)=h\}$.

\subsection {Sets $\mathcal{A}_{\ttt,p}$ of points with given approximation rates by the dyadics } \
Let $p\geq 1$, and consider $N_p$ and $\mu_{N_p}$. As usual $\oB(x,r)$ stands for the closed ball of centre $x$ and radius $r$.

\begin{definition}
Let us introduce, for every real number $\ttt \geq 1$, the set of points
$$
\mathcal{A}_{\ttt,p} =
 \bigcup_{\bfk\in  \Z_ {J_{N_p}} }
 \ \ \oB( ({\bfk }+{\bf e})2^{-J_{N_p}} , 2^{- \ttt J_{N_p }}).
$$ 
and then let us define
$$\mathcal{A}_\ttt = \bigcap_{P\geq 1} \ \bigcup_{p\geq P} \ \mathcal{A}_{\ttt,p} = \{x\in\zud: \ x \mbox{ belongs to infinitely many }\mathcal{A}_{\ttt,p}\}.$$
\end{definition}

Essentially, $\mathcal{A}_{\ttt,p}$ consists of the points of $\zud$ which are located close to the Dirac masses  of $\pi_{J_{N_p}}$ (and thus 
close to some of the Dirac masses of $\mu_{N_p}$). The larger $\ttt$, the closer $\mathcal{A}_{\ttt,p}$ to the dyadic points of generation $J_{N_p}$. 
Then $\mathcal{A}_\theta$  contains the points which are infinitely often close to some Dirac masses. 

\begin{lemma} Let $\ep>0$, then there exists an integer $p_\ep$ such that for every $p\geq p_\ep$
and for every $x\in \mathcal{A}_{\ttt,p}$
\begin{equation}\label{*7**}
 \mu  (B(x,2\cdot 2^{- \ttt J_{N_p }} ) )  \geq  2^{-d (1+ 2\ep) J_{N_p}} .
 \end{equation}
 \end{lemma}
\begin{proof}
Obviously, when $x\in \mathcal{A}_{\ttt,p}$,  the closed ball $\oB(x,2^{- \ttt J_{N_p }} )$ contains a Dirac mass of $\mu_{N_p}$ located at some  element $( {\bfk} +  {\bf e})2^{-J_{N_p}}$ (which is the location of a Dirac mass of $\pi_{J_{N_p}}$). Hence,   
$$\mu_{N_p} (\oB(x,2^{- \ttt J_{N_p }} ) ) 
\geq 2^{-J_{N_{p}}/N_{p}}2^{-dJ_{N_{p}}}=  2^{-d J_{N_{p}} (1+1/ (d N_{p}))}.$$
 
 Hence,  if $p$ is large enough to have $\eee>1/dN_{p}$, then
 \begin{equation}
 \label{ineg1}
 \mu_{N_{p}}(\oB(x,2^{- \ttt J_{N_p }} ) ) \geq      2^{-d (1+\ep) J_{N_p}} .
 \end{equation}

As in Proposition \ref{propmajorationd} we use a specific function 
$f_{\ttt,{p}}\in  \mbox{Lip}(\zud) $, that we define as follows: $f_{\ttt,{p}}(z)=2^{-\ttt J_{N_{p}}}$ for $z\in \oB(x,2^{-\ttt J_{N_{p}}})$,
and $f_{\ttt,{p}}(z) = 0 $ if $z\not\in B(x,2\cdot 2^{-\ttt J_{N_{p}}})$.
Otherwise choose an  extension of $f_{\ttt,{p}}$
onto $B(x,2\cdot 2^{-\ttt J_{N_{p}}})\sm \oB(x,2^{-\ttt J_{N_{p}}})$
such that $f_{\ttt,{p}}\in \mbox{Lip}(\zud)$ and
\begin{equation}\label{*7*}
0\leq f_{\ttt,{p}}\leq 2^{-\ttt J_{N_{p}}}.
\end{equation}
Obviously by construction we have
$$2^{-\ttt J_{N_{p}}} \mmm(B(x,2\cdot 2^{-\ttt J_{N_{p}}}))\geq
\int f_{\ttt,{p}} d\mmm.$$
Then by \eqref{ineg1}
$$\int f_{\ttt,{p}}  \, d\mmm_{N_{p}} \geq 2^{-\ttt J_{N_{p}}}2^{-d(1+\eee)J_{N_{p}}}.$$
Recall also that $\rrrr(\mu, \mu_{N_p}) \leq   2^{-(J_{N_p})^2(d+4)}$. If $p$ is large enough to have
$$\frac{1}{2}2^{-\ttt J_{N_{p}}}2^{-d(1+\eee)J_{N_{p}}}> 2^{-(J_{N_{p}})^{2}(d+4)},$$
then 
by 
\eqref{*7*}  
  we obtain (using the same argument as in Proposition \ref{propmajorationd}) that
\begin{eqnarray*}
2^{-\ttt J_{N_{p}}} \mmm(B(x,2\cdot 2^{-\ttt J_{N_{p}}})) & \geq & 
\int f_{\ttt,{p}}  \, d\mmm\\
& \geq &  \int f_{\ttt,{p}}  \ d\mu_{N_p}  - \rrrr(\mu, \mu_{N_p}) \\
& \geq & \frac{1}{2}\cdot 2^{-\ttt J_{N_{p}}}2^{-d(1+\eee)J_{N_{p}}}.
\end{eqnarray*}
This yields
$$
 \mu  (B(x,2\cdot 2^{- \ttt J_{N_p }} ) )
\geq
\frac{1}{2} 2^{-d(1+\eee)J_{N_{p}}}
>
 2^{-d (1+ 2\ep) J_{N_p}} ,
$$
 the last inequality being true when $p$ is large.
 \end{proof}

\subsection {First results on local regularity and on the size of $\mathcal{A}_{\ttt,p}$}

\begin{proposition}
\label{prop1}
If $\ttt>1$ and $x \in \mathcal{A}_\ttt$, then $h_\mu(x)\leq d/\ttt$.
\end{proposition}
\begin{proof}
Let  $x \in \mathcal{A}_\ttt$. Then \eqref{*7**} is satisfied for an infinite number of integers $p$. In other words, there is a sequence of real numbers $r_p$ decreasing to zero such that
$$ \mu(B(x,2r_p)) \geq (r_p)^{ \frac{d}{\ttt}  ( 1+ 2\ep)}.$$
This implies that $h_\mu(x) \leq \frac{d}{\ttt} ( 1+ 2\ep)$. Since this holds for any choice of $\ep>0$, the result follows.
\end{proof}

\begin{proposition}
\label{prop2}
For every $\ttt \geq 1$, $\dim_{\mathcal{H}} \mathcal{A}_\ttt \leq  d/\ttt$.
\end{proposition}
\begin{proof}
The upper bound is trivial for $\ttt=1$. Let $\ttt>1$, and let $s>d/\ttt$.

Obviously $\mathcal{A}_\ttt$ is covered by the union $ \bigcup_{p\geq P} \ \mathcal{A}_{\ttt,p}$, for any integer $P\geq 1$. Using this cover 
for large $P$'s
to bound from above the $s$-dimensional pre-measure of $\mathcal{A}_\ttt$, we find for any $\ddd>0$
\begin{eqnarray*}
\mathcal{H}^s_{\ddd }(\mathcal{A}_\ttt) & \leq & \mathcal{H}^s_{\ddd} (\bigcup_{p\geq P} \ \mathcal{A}_{\ttt,p})\\
& \leq &  \sum_{p\geq P} \  \ \  \sum_{\bfk\in   \Z_{J_{N_p}} }   
|B( ({\bfk} +{\bf e})2^{-J_{N_p}}, 2^{- \ttt J_{N_p }} )|^s \\
& \leq & C \sum_{p\geq P}  \  \ \ \sum_{\bfk\in \Z_{J_{N_p}} }  \ 
  2^{-  s\ttt J_{N_p }}   \\
& \leq & C \sum_{p\geq P}  \  2^{ dJ_{N_p}}   \  2^{-  s\ttt J_{N_p }}  ,
\end{eqnarray*}
the last sum being convergent since $s\ttt>d$. This sum converges to zero when $P\to \infty$, as a tail  of a convergent series. Hence $\mathcal{H}^s_{\ddd}(\mathcal{A}_\ttt) =0$ for every $s>d/\ttt$
and $\ddd>0$.
This implies $0=\lim_{\ddd\to0} \mathcal{H}^s_{\ddd}(\mathcal{A}_\ttt)=
\mathcal{H}^s(\mathcal{A}_\ttt)$
 and we deduce that $\dim_{\mathcal{H}} \mathcal{A}_\ttt \leq d/\ttt$.
\end{proof}

\subsection{The lower bound for the dimension of $\mathcal{A}_\ttt$}
\label{sec_cons4}

\begin{theorem}
\label{th2}
For every $\ttt > 1$, there is a measure $m_\ttt$ supported in $\mathcal{A}_\ttt $  satisfying 
\begin{equation}
\label{eq0}
\mbox{ for every Borel set $B\subset \zud$, } \ \ m_\ttt(B) \leq |B|^{d/\ttt -\psi(|B|)},
\end{equation}
where $\psi:\R^+\to\R^+$ is a gauge function, i.e. a positive continuous increasing function such that $\psi(0)=0$.

In particular, by \eqref{defexpmu}, \eqref{*dimmu} and
\eqref{eq0}, $\dim_{\mathcal{H}} m_\ttt \geq d/\ttt$.
\end{theorem}

The proof of Theorem \ref{th2} is decomposed into two lemmas. 
Essentially we apply the classical method of constructing a Cantor set  $ \calc_\ttt$  included in  $\mathcal{A}_\ttt$ and simultaneously the measure $m_\ttt$ supported by $\mathcal{A}_\ttt$ which satisfies \eqref{eq0}.

\sk

We select and fix a sufficiently rapidly growing subsequence of $(N_p)_{p\geq 1}$, that,
for ease of notation, we still denote by $(N_p)_{p \geq 1}$, such that   $N_{1}>100,$ and for  every $p\geq 1$
\begin{eqnarray}
&& J_{N_{p+1}}> \max(100\cdot \ttt J_{N_{p}}, p^{2}),\label{*142*}\\ 
\mbox{ and } &&  2^{d J_{N_{p+1}}(1-1/(p+1))}    \leq 2^{-d\ttt J_{N_{p}}}2^{dJ_{N_{p+1}}-2}. 
\label{*25b}
\label{test}
\end{eqnarray}
Since $J_{N_{p}}\to\oo$ as $N_{p}\to\oo$ it is clear
that \eqref{*142*} and \eqref{*25b} can be satisfied
by choosing a suitable subsequence.
 We will need these assumptions to ensure that the next Cantor set
 generation used during the definition
of $\cac_{\ttt}$ is much finer than the previous one and hence $\cac_{\ttt}$ is nonempty and we can use estimate \eqref{inegdelta3} later.

Moreover, we also suppose that
$N_{p}$ is increasing so rapidly 
that for $p\geq 3$
\begin{equation}\label{*p3eq}
2^{-dJ_{N_{p}}(1+1/p)}\leq 
\Big(\prod_{k=1}^{p} 2^{d J_{N_k} }   \Big) ^{  \!-1}   
\mbox{ and }
\end{equation}
$$   \Big (\prod_{k= 1}^{p}  2^{d J_{N_k} ( 1 -1/k) } \Big) ^{  \!-1} 
\leq 2^{-dJ_{N_{p}}(1-2/p)}.$$

Then, the construction of $ \calc_\ttt$ is achieved as follows:
\begin{itemize}
\item
 
The first generation of  cubes of $ \calc_\ttt$  consists of
all  the  balls of the form  
$\oB( ({\bfk} + {\bf e}) 2^{-{J_{N_1}}}, 2^{-\ttt{J_{N_1} }-1})\sse\zud$, where $\bfk \in \Z_{J_{N_1}} $.  We call $\mathcal{F}_{1}$ the set of  such  cubes, and we set $\Delta_1= \# \mathcal{F}_1$.
 Then,  a measure $m_1$ is defined as follows: for every cube  $I \in  \mathcal{F}_1$, we set 
$$m_1(I ) =\frac{1} { \Delta_1}.$$
The probability measure $m_1$ gives the same weight to each dyadic cube of first generation. 
 The measure $m_1$ can be extended to a Borel probability measure on the algebra generated by $\mathcal{F}_1$,  i.e. on  $ \sigma ( I: I\in  \mathcal{F}_1)$.

\sk

 \item
Assume that we have constructed the first $p\geq 1$ generations of cubes $\mathcal{F}_1$, $\mathcal{F}_2$, ..., $\mathcal{F}_p$ and a measure $m_p$ on the algebra 
$ \sigma \big( L: L\in  \mathcal{F}_p\big).$
 Then we choose the cubes of generation $p+1$ as those closed balls
of the form $\oB( ({\bfk} + {\bf e}) 2^{-{J_{N_{p+1}}}},2^{-\ttt J_{N_{p+1}}-1})$  which are entirely included in one (and, necessarily, in  only one)  cube $I  $ of generation $p$
and $\bfk \in \Z_{J_{N_{p+1}}} $. 
We call $\mathcal{F}_{p+1}$ the set  consisting of them. 
 We also set for every $I'\in \mathcal{F}_p$,  
 $$\Delta^{I'}_{p+1}= \# \{ I\in\mathcal{F}_{p+1} :
 I\subset I'\}.$$ 
Then we define the measure $m_{p+1}$: For every cube $I \in \mathcal{F}_{p+1}$,    we set
$$m_{p+1}(I) =  m_{p}(I')\frac{1} { \Delta^{I'}_{p+1} },$$
where $I'$ is the unique cube of generation $p$ in  $\mathcal{F}_p$ such that $I\subset I'$.\\
 
The probability measure $m_{p+1}$ can be extended to a Borel probability measure on the algebra   $ \sigma ( L: L\in  \mathcal{F}_{p+1})$ generated by $\mathcal{F}_{p+1}$. 
\end{itemize}

Finally, we set
$$ \calc_\ttt=  \bigcap_{p\geq 1} \ \ \   \bigcup_{I \in \mathcal{F}_p }  \ \  \ \ I . 
$$
By the Kolmogorov extension theorem,  $(m_p)_{p\geq 1}$ converges weakly to a Borel probability measure $m_\ttt$  supported on  $ \calc_\ttt$ and such that for every $p\geq 1$,  for every $I\in \mathcal{F}_p$, 
 $m_\ttt(I)=m_p(I)$.

\subsection{Hausdorff dimension of $ \calc_\ttt$ and $m_\ttt$}
\label{seccons2}

We first prove  that  $m_\ttt $ has  uniform behavior on the cubes belonging to  $\bigcup_{p}\mathcal{F}_p$.
 \begin{lemma}
 \label {avantdernierlemme} 
When $p$ is sufficiently large, for every cube $I\in \mathcal{F}_p$, 
\begin{equation}\label{*I*}
2^{-d J_{N_p}(1+1/p)}  \leq m_\ttt (I) \leq   2^{-dJ_{N_p}(1-2/p)}  
\end{equation}
and 
  \begin{equation}
\label{ident5}
 |I|^{\frac{ d}{ \ttt} +\frac{1}{|\log|I||} }    \leq m_\ttt (I) \leq   
 |I|^{\frac{ d}{ \ttt}- \frac{1}{|\log|I||}  } .
 \end{equation}
\end{lemma}

 \noindent {\bf Proof:} Obviously,
\begin{equation}\label{inegdelta1}
  \Delta_1  
  = 2^{d J_{N_1}}.
\end{equation}
Using $J_{N_{1}}>N_{1}>100$
\begin{equation}
\label{inegdelta1bis}
  2^{ d J_{N_1}(1-1/2)}    \leq 
  \frac{1}{2}\cdot  2^{d J_{N_1}}  \leq \Delta_1 .
  \end{equation}

Let $I$ be a cube of generation $p\geq 1 $   in the Cantor set $\mathcal{C}_\theta$.
The subcubes  of $I $ are of the form $\oB( ({\bfk} +{\bf e})2^{-J_{N_{p+1}}},2^{-\ttt  J_{N_{p+1}}-1})$ and  are regularly distributed.
Next, when calculating the number of these subcubes in $I$
on the right-handside of 
the inequality in
\eqref{*28.5} a factor $1/2$ will take care
of the fact that for a few   $(\bfk+{\bf e}) 2^{-J_{N_{p+1}}}$
on the frontier of $I$, we do not have
$\oB ( ( \bfk + {\bf e})2^{-J_{N_{p+1}}}, 2^{-\ttt J_{N_{p+1}}-1})\sse I $.
 we deduce that
\begin{eqnarray}\label{*28.5} 
  \Delta^{I}_{p+1}  & \geq &  \frac{1}{2}     ( |I|) ^d2^{ d J_{N_{p+1}} }    =  
  \frac{1}{2}
  2^{-d\ttt J_{N_{p}}}\cdot 2^{dJ_{N_{p+1}}} \\
 \nonumber  \mbox{and } \ \  \Delta^{I}_{p+1} &=& 
  \# \{ I'\in \mathcal{F}_{p+1}: I'\subset I \}  \, \leq   \,   2 (    |I|)^d  \, 2^{ d J_{N_{p+1}} } .
\end{eqnarray}
 
 Using \eqref{*25b}, 
 and the fact that $ 2( |I|)^d \leq 1$,  we  obtain
 \begin{equation}
\label{inegdelta3}
2^{d J_{N_{p+1}}(1-1/(p+1))}   \leq  \Delta^{I}_{p+1}   \leq 2^{ d J_{N_{p+1}} } .
\end{equation}

\sk

Recalling that $I\in\mathcal{F}_p$,   for $k\leq p$ denote by
$I_k$ the unique cube in $\mathcal{F}_k$ containing $I$.
Set $I_{0}=\zud$ and $\DDD^{I_{0}}=\DDD_{1}$.
We obtain
\begin{equation}\label{*18*}
  \big (\prod_{k=1}^{p}  \Delta^{I_{k-1}}_{k} \big) ^{-1}   =m_\ttt (I).
 \end{equation}

The key property is that in   \eqref{inegdelta3}, the bounds are uniform in $I\in \mathcal{F}_p$. Hence,  
$$  \Big(\prod_{k=1}^{p} 2^{d J_{N_k} }   \Big) ^{  \!-1}    \leq m_\ttt (I) \leq  
    \Big (\prod_{k= 1}^{p}  2^{d J_{N_k} ( 1 -1/k) } \Big) ^{  \!-1}    .$$
By \eqref{*p3eq} 
we have \eqref{*I*} when $p\geq 3$.

This means that, the measure $m_\ttt $ is almost uniformly distributed on the cubes of 
the same generation. Since these cubes $I\in\mathcal{F}_p$ have 
 the same diameter which is  $ 2^{-\ttt J_{N_p}}$, we  obtain 
  for large $p$'s that
  \begin{equation}\label{*19**}
 |I|^{\frac{ d}{ \ttt}(1+1/p)}    \leq m_\ttt (I) \leq    |I|^{\frac{ d}{ \ttt}(1- 2/p)}.
 \end{equation}

 Finally, we remark that  by \eqref{*142*},
 $p =o(|\log |I||)$  when $I\in\mathcal{F}_p$ is arbitrary,  
 hence \eqref{*19**} yields \eqref{ident5}. \hfill $\Box$
 
\mk
%

Now we extend \eqref{ident5} and Lemma \ref{avantdernierlemme} to all Borel subsets of $\zu$.

\begin{lemma}
\label{lastlemma}
There is a continuous  increasing  mapping $\psi:\R^+\to \R^+$, satisfying $\psi(0)=0$, and there is  $\eta >0$, such that for any Borel set $B\subset \zu$ with $|B|<\eta$ we
have  
 \begin{equation}
\label{ident6}
  m_\ttt (B) \leq    |B|^{\frac{d}{\ttt}- \psi(|B|)} .
 \end{equation}
 \end{lemma}

%
 
\begin{proof} 
Fix $\ep_1= 2^{-1}$, 
a Borel set
$B\subset [0,1]$ with $ |B|<2^{-J_{N_{1}}}=\eta_{0}.$ Let $p\geq 2$ be the unique integer such that \begin{equation}\label{*pot1*}
 2^{-J_{N_p}}\leq |B| < 2^{-J_{N_{p-1}}}.
\end{equation}

Let us distinguish two cases:

\mk

$\bullet$ {\bf  $\mathbf {  2^{-\ttt J_{N_{p-1}}} \leq |B| <  2^{- J_{N_{p-1}}}:}$} By \eqref{*pot1*},   $B$   intersects at  most $2^d$ cubes  $I'\in \mathcal{F}_{p-1}$. 
If there is no such cube then
$m_\ttt(B)=0$. Otherwise,
 denoting by $I'$ one of these cubes,
using  \eqref{*I*} and  \eqref{*19**} we find that
\begin{eqnarray*}
m_\ttt(B) &  \leq  & 2^d \cdot  m_\ttt (I')   \,  \leq \,  2^{d}\cdot  
2^{- d J_{N_{p-1}}(1-\frac{2}{p-1})}\\
&\leq &  C \cdot |B|^{  \frac{d}{\ttt}( 1-\frac{2}{p-1})}<|B|^{\frac{d}{\ttt}-\ep_1}.
\end{eqnarray*}
when  $p$ is sufficiently large.  Recall that $p$ is related to the diameter of $B$: 
the smaller $|B|$ is,
the larger $p$ becomes.

\mk

$\bullet$  {\bf  $\mathbf {2^{-J_{N_p}}<|B|< 2^{-\ttt J_{N_{p-1}}}}:$} We will 
determine a sufficiently small $\eta_{1}\in (0,\eta_{0})$ later and will 
suppose that $|B|<\eta_{1}$. For small $|B|$'s, that is,
for large $p$'s, $B$ intersects at most one cube  $I'\in \mathcal{F}_{p-1}$. If there is no such cube then
$m_\ttt(B)=0$. Hence we need to deal with the case when such a cube $I'$ exists. Obviously,  $|B|<|I'|
=  2^{-\ttt J_{N_{p-1}}}$.
The mass $m_\ttt(I')$
is distributed evenly on the cubes
$\oB( (\bfk  + {\bf e})2^{-J_{N_{p}}},2^{-\ttt J_{N_{p}}-1})\sse I'.$
By \eqref{*142*}, $J_{N_{p}}>100\cdot \ttt J_{N_{p-1}}$.
On one hand,  we saw 
in \eqref{*28.5}
that $\DDD_{p}^{I'}>\frac{1}{2}  2^{-d\ttt J_{N_{p-1}}}2^{dJ_{N_{p}}}$. 
We deduce that the mass of a ball $I$ in $\mathcal{F}_{p}$ included in $I'$ has  $m_\theta$-mass which satisfies
\begin{equation}
\label{eqfinalm}
 m_\theta(I) =  m_\theta(I') \frac{1}{\DDD_{p}^{I'}} \leq  m_\theta(I')   2^{1+d\ttt J_{N_{p-1}} -dJ_{N_{p}}}.
\end{equation}
On the other hand,  since $B$ is within a cube of side length $2|B|$
the number of 
cubes of generation $p$ (i.e. of the form $\oB( (\bfk + {\bf e})2^{-J_{N_{p}}},2^{-\ttt J_{N_{p}}-1})$)  intersecting $B$ is less than $4^{d}|B|^{d}2^{dJ_{N_{p}}}.$

Hence, combining \eqref{*I*}, \eqref{eqfinalm} and $|B|^{-1/\ttt}>2^{J_{N_{p-1}}}$, we find that
\begin{eqnarray*}
m_\ttt(B) & \leq  &  \sum_{I\in \mathcal{F}_p: I \cap B \neq \emptyset}  \ \ \ m_\ttt (I) \\
 & \leq   & 4^{d}|B|^{d}2^{dJ_{N_{p}}}  m_\theta(I') 
 2^{1+d\ttt J_{N_{p-1}} -dJ_{N_{p}}}  \leq  C  |B|^{d}   2^{ d\ttt J_{N_{p-1}}}   m_\theta(I')    \\
  &  \leq  &   C  |B|^{d}  2^{ d\ttt J_{N_{p-1}}}  2^{-dJ_{N_{p-1}}(1-\frac{2}{p-1})} \leq C  |B|^{d}   \cdot 2^{d  J_{N_{p-1}}(\ttt-1+\frac{2}{p-1})}   \\
  &\leq &   
 C |B|^{d} \cdot |B|^{-\frac{1}{\ttt}d(\ttt-1+\frac{2}{p-1})}  \leq     |B|^{\frac{d}{\ttt}(1-\frac{2}{p-1})} \leq  |B|^{\frac{d}{\ttt}-\ep_1},
\end{eqnarray*}
the last inequality being true  for large $p$, i.e. for Borel sets $B$ of diameter small enough (by the same argument as above). 

 We can thus choose  $\eta_{1}\in (0,\eta_0)$ so that  
 \begin{eqnarray*}
 \mbox{when $|B|\leq \eta_{1}$,  } \ \  m_\ttt(B)    \leq      |B|^{\frac{d}{\ttt}-\ep_1}.
\end{eqnarray*}

\sk

Fix now $\ep_2=2^{-2}$. By the same method as above, we find $0<\eta_2<\eta_{1}$ 
such that if $|B|\leq \eta_2$, 
$$m_\ttt(B)\leq  |B|^{\frac{d}{\ttt}-\ep_2}.$$

\sk

We iterate the procedure:  $\forall \ p> 1$, there is 
$0< \eta_{p}<\eta_{p-1}$ such that 
$$\mbox{if  $|B|\leq \eta_p$,   } \ \ m_\ttt(B)\leq  |B|^{\frac{d}{\ttt}-\ep_p}, \ \ \ \  \mbox{ where $\ep_p=2^{-p}$.}$$

\sk

In order to conclude, we consider a map $\psi$  built as an increasing continuous interpolation function which goes through the points $(\eta_{p+1},\ep_p)_{p\geq 1}$
and $(\eta_{1},\ep _{1})$. The shift in the indices in the sequence is introduced so that 
$\ep _{p}\leq \psi(x)\leq \ep _{p-1}$ holds for $x\in [\eta_{p+1},\eta_{p}]$. Hence \eqref{ident6} holds true 
for every Borel set  $B$ satisfying $|B| \leq \eta:=\eta_1$. 
\end{proof}


\mk

We can now conclude our results on the values of the spectrum of $\mu$.

\begin{proposition}
\label{propfinal}
For any $h\in [0,d)$, $d_\mu(h)=h$.
\end{proposition}

\begin{proof}

Let $h\in (0,d)$, and let $\ttt =d/h >1$. Recall that
$$\widetilde {E_\mu}(h) = \{  x\in\zud: h_\mu(x)\leq h\} = \bigcup_{h'\leq h} \ \ E_\mu(h').$$

By  Proposition \ref{prop1}, $\mathcal{A}_\ttt \subset\widetilde {E_\mu}(h)$, hence by Proposition \ref{prop0} we have $\dim \mathcal{A}_\theta \leq \dim \widetilde {E_\mu}(h) \leq h$.

\sk

 Let us write
$$\mathcal{A}_\ttt =  \Big(\mathcal{A}_\ttt\cap E_\mu(h) \Big) \ \bigcup   \ \Big( \bigcup_{n\geq 1} \mathcal{A}_\ttt \cap \widetilde {E_\mu}(h-1/n) \Big).$$

\sk

Now, consider the measure $m_\ttt$ provided by Theorem \ref{th2}.  Since the Cantor set  $\mathcal{C}_\ttt$ is  the support of $m_\ttt $ and since it is included in $\mathcal{A}_\theta$, we have $m_\ttt (\mathcal{A}_\ttt) \geq m_\ttt (\calc_\ttt)>0$.

For any $n\geq 1$, $\dim_{\mathcal{H}} \, (\mathcal{A}_\ttt \cap \widetilde {E_\mu}(h -1/n)) \leq \dim_{\mathcal{H}} \widetilde {E_\mu} (h-1/n) \leq h-1/n <h=d/\theta$ again by Proposition \ref{prop0}.  Since we proved 
in Theorem \ref{th2}
that 
$\dim_{\mathcal{H}} m_\ttt \geq d/\ttt$, by Property \eqref{property}
we deduce that $m_\ttt (\mathcal{A}_\ttt \cap \widetilde {E_\mu}(h-1/n)) =0$.

Combining the above, we see that $m_\ttt(\mathcal{A}_\ttt) = m_\ttt \Big(\mathcal{A}_\ttt\cap E_\mu(h) \Big) >0$, hence $\dim_{\mathcal{H}}   E_\mu(h)  \geq \dim_{\mathcal{H}} \mathcal{A}_\ttt\cap E_\mu(h)  \geq d/\ttt =h $. We already had   the corresponding  upper bound, 
hence the result for $h\in (0,d)$.

\mk

It remains us to treat the case  $h=0$.

\mk

For $h=0$, we consider the set $\mathcal{A}_\infty = \bigcap_{\ttt>1} \mathcal{A}_\ttt  = \bigcap_{N\geq 1} \mathcal{A}_{N}$.  This set is non-empty and uncountable, 
since each $\caa_{N} $
contains a dense $G_\de$ subset of $\zu^d$,
namely 
$\cap_{P\geq 1}\cup _{p\geq P}\intt(\caa_{N,p})$
and the countable intersection of dense $G_{\ddd}$
sets is still dense $G_{\ddd}$.
 Moreover, by Proposition \ref {prop1}, every $x\in \mathcal{A}_\infty$ has exponent 0 for $\mu$. Hence, $\mathcal{A}_\infty \subset E_\mu(0)$ and $\dim_{\mathcal{H}} E_\mu(0) =0$. 
\end{proof}


\end{document}